\PassOptionsToPackage{
  colorlinks=true,
  linkcolor=blue,
  citecolor=blue,
  urlcolor=blue,
  hyperfootnotes=false
}{hyperref}

\documentclass[submission,copyright,creativecommons,sharealike, noncommercial]{eptcs}

\providecommand{\event}{GandALF 2026}

\usepackage{placeins}
\usepackage{xr}
\usepackage{pdfpages}
\usepackage{xcolor}

\hypersetup{
  colorlinks=true,
  linkcolor=blue,
  citecolor=blue,
  urlcolor=blue
}

\usepackage{iftex}
\usepackage{subcaption} 
\usepackage{subfiles}

\ifpdf
  \usepackage{underscore}
  \usepackage[T1]{fontenc}
\else
  \usepackage{breakurl}
\fi

\usepackage[fleqn]{amsmath}
\usepackage{amssymb}
\usepackage{latexsym}
\usepackage{amsfonts}
\usepackage{mathtools}
\usepackage{amsthm}
\usepackage{aliascnt}
\usepackage{microtype}
\microtypesetup{nopatch=footnote}

\usepackage[all]{nowidow}

\newtheoremstyle{break}
  {}
  {}
  {}
  {}
  {\bfseries}
  {.}
  {\newline}
  {}

\theoremstyle{break}

\newtheorem{theorem}{Theorem}[section]

\newaliascnt{definition}{theorem}
\newtheorem{definition}[definition]{Definition}
\aliascntresetthe{definition}

\newaliascnt{lemma}{theorem}
\newtheorem{lemma}[lemma]{Lemma}
\aliascntresetthe{lemma}

\newaliascnt{corollary}{theorem}

\aliascntresetthe{corollary}

\newaliascnt{observation}{theorem}

\aliascntresetthe{observation}

\newaliascnt{remark}{theorem}

\aliascntresetthe{remark}

\definecolor{dgreen}{RGB}{5,150,25}
\definecolor{red40}{RGB}{200,20,25}

\usepackage{graphicx}
\usepackage{float}
\usepackage{setspace}

\usepackage{tikz}
\usetikzlibrary{automata,positioning,arrows.meta, calc}

\usepackage[nameinlink,noabbrev]{cleveref}

\crefname{theorem}{Theorem}{Theorems}
\Crefname{theorem}{Theorem}{Theorems}

\crefname{definition}{Definition}{Definitions}
\Crefname{definition}{Definition}{Definitions}

\crefname{lemma}{Lemma}{Lemmas}
\Crefname{lemma}{Lemma}{Lemmas}

\crefname{corollary}{Corollary}{Corollaries}
\Crefname{corollary}{Corollary}{Corollaries}

\crefname{observation}{Observation}{Observations}
\Crefname{observation}{Observation}{Observations}

\crefname{remark}{Remark}{Remarks}
\Crefname{remark}{Remark}{Remarks}

\crefname{conjecture}{Conjecture}{Conjectures}
\Crefname{conjecture}{Conjecture}{Conjectures}

\crefname{algorithm}{Algorithm}{Algorithms}
\Crefname{algorithm}{Algorithm}{Algorithms}

\crefname{figure}{Figure}{Figures}
\Crefname{figure}{Figure}{Figures}

\title{Every Regular Language Has a\\ Bipartite Cactus Automaton}
\author{
Mario Grobler
\institute{University of Bremen\\ Bremen, Germany}
\email{\quad grobler@uni-bremen.de}
\and
Till Koziel
\institute{RWTH Aachen University\\ Aachen, Germany}
\email{\quad till.koziel@rwth-aachen.de}
\and
Sebastian Siebertz
\institute{University of Bremen\\ Bremen, Germany}
\email{\quad siebertz@uni-bremen.de}
}
\def\titlerunning{Every Regular Language Has a Bipartite Cactus Automaton}
\def\authorrunning{M.\ Grobler, T.\ Koziel, S.\ Siebertz}

\begin{document}
\maketitle

\begin{abstract}
Thompson's classical construction turns every regular expression into an equivalent $\varepsilon$-NFA whose transition graph is series-parallel.
The construction uses $\varepsilon$-transitions, and eliminating them by the usual shortcut construction need not preserve planarity or the bound of two on the treewidth of the transition graph.
Book and Chandra proved that every regular language nevertheless has an $\varepsilon$-free NFA with a planar transition graph.
We show that the transition graph can in fact be chosen to be a bipartite cactus: every cycle is even, and any two cycles have at most one vertex in common.
In particular, it is outerplanar and has treewidth at most two.
This is a considerably more simple structure than the general series-parallel guarantee supplied by Thompson's construction.
The proof is a direct induction on parity-typed regular expressions and uses only elementary closure properties of cactus graphs.
\end{abstract}

\section{Introduction}\label{sec:intro}
Nondeterministic finite-state automata (NFAs) are central objects of study in theoretical computer science, see, for example, \cite{Eilenberg1974,HopcroftMotwaniUllman2006,Kleene1956,Kozen1997,RabinScott1959}.
Although they are usually studied through the languages they recognize, the graph structure of their state diagrams is also a natural object of study.
A classical example is Eggan's theorem, which relates the star height of a regular language to the cycle rank of directed transition graphs of automata recognizing it~\cite{Eggan1963}.

Very simple transition graphs arise naturally from Thompson's construction for regular expressions~\cite{Thompson1968RegularExpressionSearch}.
Given a regular expression, the construction produces an equivalent $\varepsilon$-NFA whose transition graph is series-parallel.
This gives a simple structural representation of regular languages, but it relies on $\varepsilon$-transitions.
The usual elimination procedure introduces labelled shortcut transitions along $\varepsilon$-paths, and this operation need not preserve the series-parallel structure, planarity, or the treewidth bound of two.

The graphs obtained in this paper are considerably more restricted than general series-parallel graphs.
A connected graph is a \emph{cactus} if every edge belongs to at most one cycle, equivalently, every block is either a single edge or a cycle.
The cacti form a proper subclass of the outerplanar graphs, and every outerplanar graph is a subgraph of a series-parallel graph, see~\cite{brandstadt1999graph}. 

\begin{figure}[t]
    \centering
    \begin{subfigure}{.40\textwidth}
    \centering
    \begin{tikzpicture}[shorten >=1pt, node distance=1.5cm, on grid, auto]
        \node[state, initial, minimum size=.6cm]   (s)         {$s$};
        \node[state, minimum size=.6cm]            (q2)    [right = of s]    {};
        \node[state, minimum size=.6cm] at (1.5,1.25)           (q1)      {};
        \node[state, minimum size=.6cm] at (1.5,-1.25)            (q3)        {};
        \node[state, accepting, minimum size=.6cm] (f)     [right = of q2] {};

        \path[->]
            (s)     edge    node    {a} (q1)
            (s)     edge    node    {b} (q2)
            (s)     edge    node    {c} (q3)
            (q1)    edge    node    {a} (f)
            (q2)    edge    node    {b} (f)
            (q3)    edge    node    {c} (f)
            ;
    \end{tikzpicture}
    \caption{A natural NFA whose transition graph is planar, but not a cactus}
    \end{subfigure}\hspace{5mm}
    \begin{subfigure}{.40\textwidth}
    \centering
        \begin{tikzpicture}[shorten >=1pt, node distance=1.5cm, on grid, auto]
        \node[state, initial, minimum size=.6cm]   (s)         {$s$};
        \node[state, minimum size=.6cm]            (q2)    [right = of s]    {};
        \node[state, minimum size=.6cm] at (1.5,1.25)           (q1)      {};
        \node[state, accepting, minimum size=.6cm] (q5) [right = of q2] {};
        \node[state, accepting, minimum size=.6cm] (q4) [right = of q1] {};
        \node[state, minimum size=.6cm] at (1.5,-1.25)            (q3)        {};
        \node[state, accepting, minimum size=.6cm] (q6) [right = of q3] {};

        \path[->]
            (s)     edge    node    {a} (q1)
            (s)     edge    node    {b} (q2)
            (s)     edge    node    {c} (q3)
            (q1)    edge    node    {a} (q4)
            (q2)    edge    node    {b} (q5)
            (q3)    edge    node    {c} (q6)
            ;
    \end{tikzpicture}
    \caption{An equivalent NFA whose transition graph is a tree}
    \end{subfigure}
    \caption{Two $\varepsilon$-free NFAs recognizing $\{aa,bb,cc\}$}
    \label{fig:cactus-example}
\end{figure}
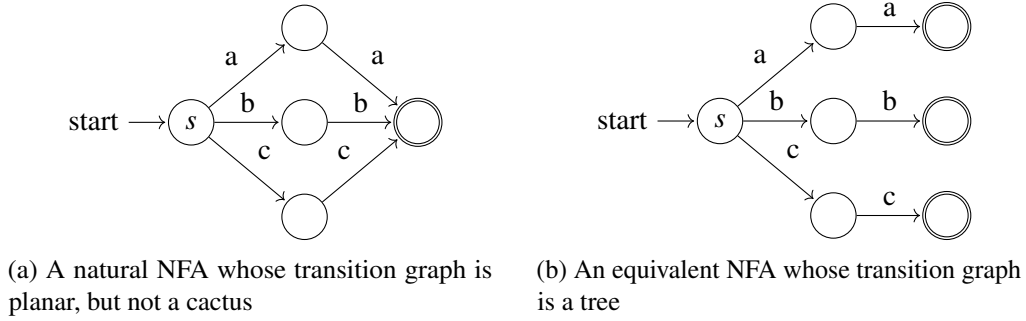

In \Cref{fig:cactus-example}, the left automaton is a natural $\varepsilon$-free NFA for $\{aa,bb,cc\}$ whose transition graph is~$K_{2,3}$. This graph is series-parallel, but it is not a cactus: every edge belongs to two cycles. The right automaton is an equivalent $\varepsilon$-free NFA whose transition graph is a tree. The example illustrates how nondeterminism and state duplication can be used to improve the graph structure without changing the accepted language.

Book and Chandra showed that planarity can still be achieved when $\varepsilon$-transitions are not allowed~\cite{BookChandra1976}.
They proved that every regular language is recognized by an $\varepsilon$-free NFA with a planar transition graph.

The corresponding statement fails for deterministic automata~\cite{BookChandra1976}.
This observation initiated a line of work on topological restrictions of automata.
Subsequent work studied planar finite automata~\cite{BezakovaPal1999}, finite automata with restricted state graphs~\cite{KutribMalcherSchneider2022}, and, in the deterministic setting, the genus of regular languages~\cite{BonfanteDeloup2018,BonfanteDeloup2019,BonfanteDeloup2024}.

\medskip
We strengthen the results of Thompson and of Book
and Chandra by proving the following theorem.

\begin{theorem}\label{the:main}
For every regular language, there exists an $\varepsilon$-free NFA that recognizes the language and has a bipartite cactus transition graph.
\end{theorem}

\medskip
Like Thompson's proof, our proof proceeds by structural induction over a regular expression.
However, we first decompose the expression into its even and odd parts, in such a way that every Kleene star is applied to an even language.
We then maintain an explicit bipartition throughout the construction.
After deleting the initial state, the automaton consists of independent components with exactly one entry each.
Every non-initial source of a transition into an accepting state is joined to the initial state by a unique path in the corresponding local graph.
This simple invariant is what we use for the Kleene-star construction: adding an edge between the ends of such a path creates at most one new cycle and cannot make two cycles overlap.
Our constructions use state duplication to avoid the introduction of $\varepsilon$-transitions.
The graph-theoretic part of the proof consequently uses only that cacti are preserved when they are joined at a single vertex or by a bridge, and when an edge is added between two vertices joined by a unique path.

Our theorem has the nice consequence that Courcelle's theorem applies: every fixed $\mathrm{MSO}_2$-definable property of the transition structure can be decided in linear time, while standard extensions cover broad classes of counting and optimization problems~\cite{Courcelle1990MSO,ArnborgLagergrenSeese1991}.
Furthermore, many problems can additionally be solved with constant-workspace in polynomial time algorithms on cactus graphs, see~\cite{BRIMKOV2017393}. 

\section{Preliminaries}\label{sec:defs}

In this section, we establish the definitions and notation used throughout the
paper and recall the graph-theoretic facts needed in the constructions. We
assume basic familiarity with graph theory and finite automata and refer to the
literature for more background.

An \emph{alphabet} $\Sigma$ is a finite set of symbols. A \emph{word} over
$\Sigma$ is a finite sequence $w=a_1a_2\cdots a_n$ with $a_i\in\Sigma$ for
every $i\in\{1,\ldots,n\}$; the case $n=0$ gives the empty word
$\varepsilon$. The set of all words over $\Sigma$ is denoted by~$\Sigma^*$.
A \emph{nondeterministic finite-state automaton~(NFA)} is
a tuple $A=(Q,\Sigma,\Delta,s,F)$, where $Q$ is a finite set of states,
$\Delta\subseteq Q\times\Sigma\times Q$ is a transition relation, $s\in Q$
is the initial state, and $F\subseteq Q$ is the set of accepting states. For
states $p,q\in Q$, we write
$\Lambda_A(p,q)=\{a\in\Sigma\mid(p,a,q)\in\Delta\}$ for the set of labels on
transitions from $p$ to $q$.

A \emph{run} of $A=(Q,\Sigma,\Delta,s,F)$ on a word
$w=a_1a_2\cdots a_n\in\Sigma^*$ is a sequence of states
$q_0q_1\cdots q_n$ such that $q_0=s$ and
$(q_{i-1},a_i,q_i)\in\Delta$ for every $i\in\{1,\ldots,n\}$. The run is
\emph{accepting} if $q_n\in F$. The language accepted by $A$ is
$L(A)=\{w\in\Sigma^*\mid A\text{ has an accepting run on }w\}$. The
languages accepted by NFAs are exactly the \emph{regular languages}.

\medskip
The \emph{regular expressions} over $\Sigma$ are defined inductively. The
expressions $\varnothing$, $\varepsilon$, and $a$, for every $a\in\Sigma$,
are regular expressions, and if $r_1$ and $r_2$ are regular expressions, then
so are $r_1r_2$, $r_1\cup r_2$, and $r_1^*$. Their languages are given by
$L(\varnothing)=\varnothing$, $L(\varepsilon)=\{\varepsilon\}$,
$L(a)=\{a\}$, $L(r_1r_2)=L(r_1)L(r_2)$,
$L(r_1\cup r_2)=L(r_1)\cup L(r_2)$, and $L(r^*)=L(r)^*$. The languages
defined by regular expressions are precisely the regular languages.

For $p\in\{0,1\}$, a language $L\subseteq\Sigma^*$ has \emph{parity $p$} if
$|w|\equiv p\pmod 2$ for every $w\in L$. The empty language has either
parity. We use addition of parities modulo~$2$.

\medskip
All graphs in this paper are finite, simple, and undirected. For distinct
vertices $u,v$ of a graph $G$, we write $G+\{u,v\}$ for the graph obtained by
adding the edge $\{u,v\}$ if it is not already present. For $X\subseteq V(G)$,
we denote by $G[X]$ the subgraph induced by~$X$. For an edge $e$, we write
$G-e$ for the graph obtained by removing~$e$, and for $v\in V(G)$, the graph
obtained by deleting $v$ and all edges incident with it is denoted by $G-v$.
An edge $e$ of a graph $G$ is a \emph{bridge} if deleting $e$ increases the
number of connected components. A vertex is a \emph{cut vertex} if deleting
it increases the number of connected components. A \emph{block} is a maximal
connected subgraph that has no cut vertex of its own. 
In particular, every
bridge together with its endpoints forms a block, and every isolated vertex
forms a block. A connected graph is a \emph{cactus} if every edge belongs to
at most one cycle. Equivalently, any two cycles of a cactus have at most one vertex in common.
Another equivalent characterization is that every block containing an edge
is either a single edge or a cycle.

\medskip
We use the following elementary closure properties of cacti.

\begin{lemma}\label{lem:cactus-operations}
The following operations preserve cacti.
\begin{enumerate}
    \item Take cacti that are otherwise vertex-disjoint, choose one vertex
    in each, and identify all chosen vertices to a single common vertex.
    \item Take two vertex-disjoint cacti and add one edge with one endpoint in
    each of them.
    \item Let $X$ be a cactus and let $x,y\in V(X)$ be distinct. If $X$ has a unique
    $x$-$y$ path, then $X+\{x,y\}$ is a cactus.
\end{enumerate}
\end{lemma}

\begin{proof}
For the first statement, every cycle remains inside one of the original
graphs, because a cycle cannot pass through the identified vertex twice. The
edge added in the second statement is a bridge and therefore belongs to no
cycle.

For the third statement there is nothing to prove if $\{x,y\}$ is already an
edge. Otherwise, let $P$ be the unique $x$-$y$ path in~$X$. No edge of $P$
belongs to a cycle of~$X$: such a cycle would provide an alternative path
between the endpoints of that edge, and hence an alternative $x$-$y$ path.
The only new cycle in $X+\{x,y\}$ is therefore $P+\{x,y\}$, and none of its
old edges belonged to a cycle before. Thus every edge of the new graph still
belongs to at most one cycle.
\end{proof}

For an automaton $A=(Q,\Sigma,\Delta,s,F)$, its \emph{transition graph} is
$G(A)=(Q,E)$, where $\{p,q\}\in E$ if and only if $q\ne p$ and there exists $a\in \Sigma$ such that either $(q,a,p)\in \Delta$ or $(p,a,q)\in \Delta$. 
Thus directions and labels are forgotten, we have no multi-edges and and loops give no edge. 
This is natural, as it is common to depict multiple transitions between the same states that have the same direction as a single arc with multiple labels and, as we will show below, our automata will not have loops. 

Graph properties of an automaton
always refer to $G(A)$, and connected components of an automaton mean
connected components of $G(A)$. We will only consider automata with connected
transition graphs and no loops, i.e. transitions of the form $(q,a,q)$. 

A map $\chi:Q\to\{0,1\}$ is a \emph{transition-compatible coloring} of $A$
if $\chi(p)\neq\chi(q)$ for every transition $(p,a,q)\in\Delta$. In
particular, an automaton with such a coloring has no loops, and $\chi$ is a
proper $2$-coloring of its transition graph. If $\chi(s)=0$, then every run
of length $n$ ends in a state of color $n\bmod 2$. Every coloring discussed in the sequel will be transition compatible, in particular, at no point loops are introduced.

\section{From regular expressions to cactus automata}\label{sec:mainTheo}

In this section, we prove \cref{the:main}. We first separate a regular
expression into its even and odd parts. 
This is a simple special case of typed regular expressions with types in a finite monoid, which were considered e.g.\ by
Bojańczyk~\cite[Section~2.1]{Bojanczyk2008Common}. 
The parity typing below is
their specialization to the length morphism $w\mapsto |w|\bmod 2$.
What we call parity normal form is a specialization of
\cite[Theorem~3]{Bojanczyk2008Common}, applied to the product of a morphism
recognizing $L(r)$ and the length-parity morphism. We give a proof of the normal form lemma for completeness and for the graph theorists among us. 
We can then carry a bipartition through the whole construction. The cactus
property is maintained directly using the elementary operations from
\cref{lem:cactus-operations}.

\subsection{Parity normal form}

We call a regular expression \emph{typed with parity $p\in\{0,1\}$} if it is
obtained by the following rules. The expression $\varnothing$ may be given
either parity, $\varepsilon$ has parity~$0$, and every letter has parity~$1$.
The union of two expressions of parity $p$ has parity $p$, the concatenation
of expressions of parities $p$ and $q$ has parity $p+q$, and the star of an
expression of parity~$0$ has parity~$0$. Here and below, parities are added
modulo~$2$. An immediate induction shows that the language of an expression
typed with parity $p$ has parity~$p$.

\begin{lemma}[Parity normal form]\label{lem:parity-normal-form}
For every regular expression $r$, there are regular expressions $r_0$ and
$r_1$, typed with parity $0$ and $1$, respectively, such that
$L(r_p)=\{w\in L(r)\mid |w|\equiv p\pmod 2\}$ for $p\in\{0,1\}$.
Moreover, every star occurring in $r_0$ or $r_1$ is applied to an expression
of parity~$0$.
\end{lemma}

\begin{proof}
We proceed by induction over~$r$. For the atomic expressions, put
\begin{align*}
 (\varnothing)_0&=\varnothing,&(\varnothing)_1&=\varnothing,\\
 (\varepsilon)_0&=\varepsilon,&(\varepsilon)_1&=\varnothing,\\
 (a)_0&=\varnothing,&(a)_1&=a.
\end{align*}
For union, define $(u\cup v)_p=u_p\cup v_p$. For concatenation, define
\begin{align*}
 (uv)_0&=(u_0v_0)\cup(u_1v_1),&
 (uv)_1&=(u_0v_1)\cup(u_1v_0).
\end{align*}
These expressions clearly have the asserted types and languages.

It remains to consider $r=u^*$. Put $d=(u_0)^*$ and
$c=u_1d u_1d$, and define
\begin{align*}
 (u^*)_0&=d c^*,& (u^*)_1&=d u_1d c^*.
\end{align*}
Both $u_0$ and $c$ have parity~$0$, so all new stars are applied to even
expressions. To verify the languages, fix a factorization of a word in
$L(u)^*$ into factors from $L(u)$ and classify every factor as even or odd.
Grouping consecutive even factors gives a factorization
$e_0o_1e_1\cdots o_ke_k$, where $e_i\in L(d)$ and
$o_i\in L(u_1)$. If~$k$ is even, pair $o_1,o_2$, then $o_3,o_4$, and so on.
The resulting factorization belongs to $L(dc^*)$. If $k$ is odd, leave $o_1$
unpaired and pair $o_2,o_3$, then $o_4,o_5$, and so on. 
The factorization
belongs to $L(du_1dc^*)$. Conversely, every word produced by either
expression is a product of factors from $L(u)$, and the displayed pairing
shows that the number of odd factors is even in the first expression and odd
in the second. These are precisely the even and odd words in $L(u)^*$.
\end{proof}

Throughout the automata constructions below, when we speak of copies, sometimes indexed by a specified symbol, we always mean fresh, pairwise disjoint, and disjoint from every newly introduced state, unless we explicitly identify specified states. 
In particular, if $y,z$ are indices and $X$ is some set of states, then $X^y$ and $X^z$ denote copies of $X$, where 
$x^i\in X^i$ is the copy of $x\in X$, for $i\in\{y,z\}$, with $X^i \cap X = \varnothing$. Moreover, $X^z\cap X^y = \varnothing$ if and only if $z\ne y$. If no explicit index is mentioned, then it shall be implied that a the copies are labeled by a unique but otherwise unspecified index.
Each constructed automaton contains exactly the states and transitions that are stated in its construction, and a copied
state is accepting exactly when its accepting mark is explicitly retained.

We now inductively construct automata for typed expressions. For an automaton
$A=(Q,\Sigma,\Delta,s,F)$, we define the possible labels of a last transition
at a state $g\in Q$ by
\mbox{$\Gamma_A(g)=\{a\in\Sigma\mid(g,a,f)\in\Delta$} \mbox{$\text{ for some }f\in F\}$}.

To guarantee the stronger graph structure, we introduce \emph{strongly cactus
automata}. The initial state~$s$ separates the transition graph into rooted
components. Each component has exactly one entry state, contains at least one
accepting state, and has no transition back to~$s$. The additional invariant is simple: every non-initial possible source of a last transition is joined to
$s$ by a unique path in its local graph.
\medskip

\begin{definition}[Strongly cactus automaton of parity $p$]
\label{def:compAuto}
Let $p\in\{0,1\}$, let $A=(Q,\Sigma,\Delta,s,F)$ be a connected NFA, and
let $C_1,\ldots,C_n$ be the connected components of $G(A)-s$. Put
$Q_i=V(C_i)$. We call $A$ \emph{strongly cactus of parity $p$} if it
has a coloring $\chi_A:Q\to\{0,1\}$ such that the following conditions hold.
\begin{enumerate}
    \item The coloring $\chi_A$ is transition-compatible,
    $\chi_A(s)=0$, and $\chi_A(f)=p$ for every $f\in F$.
    \item The graph $G(A)$ is a cactus, and no transition enters~$s$.
    \item For every $i\in\{1,\ldots,n\}$, the set
    $\{p\in Q_i\mid (s,a,p)\in \Delta \text{ for some } a\in \Sigma\}$ consists of one state $t_i$, called the \emph{entry
    state} of $C_i$, and $Q_i\cap F\neq\varnothing$.
    \item For every $g\in Q_i$ with $\Gamma_A(g)\neq\varnothing$, the local
    graph $G(A)[Q_i\cup\{s\}]$ has a unique $s$-$g$ path.
\end{enumerate}
Transitions $(s,a,t_i)$ are the \emph{ingoing transitions} of $C_i$, and
transitions with both endpoints in $Q_i$ are its \emph{internal transitions}.
The initial state itself may be accepting.
\end{definition}

For a component $C_i$ of $G(A)-s$, we call
$G(A)[Q_i\cup\{s\}]$ its \emph{local graph}. Thus two local graphs of the
same automaton have only the initial state in common.

By the first condition, $G(A)$ is bipartite, and every word in $L(A)$ has
parity~$p$. By the second condition, its transition graph is a cactus. We call the vertices
in Condition~4 \emph{exposed}. Notice that the unique path from $s$ to
an exposed vertex consists entirely of bridges of the local graph.

\subsection{The atomic expressions}

The atomic typed expressions easily translate into strongly cactus automata.

\begin{lemma}[Automata for the atomic languages]\label{lem:atomic}
The empty language has a strongly cactus automaton of either parity.
There are strongly cactus automata of parity $0$ and $1$ accepting
$\{\varepsilon\}$ and $\{a\}$, respectively, for every $a\in\Sigma$.
\end{lemma}

\begin{proof}
The following three automata recognize $\varnothing$, $\{\varepsilon\}$, and
$\{a\}$, respectively.
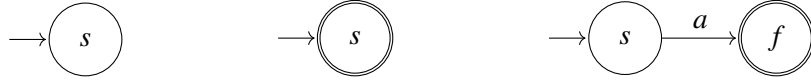
\begin{figure}[H]
  \centering
    \begin{tikzpicture}[shorten >=1pt,node distance=2cm,on grid,auto, initial text=]
      \node[state, initial] (s) {$s$};
    \end{tikzpicture}\hspace{17mm}
    \begin{tikzpicture}[shorten >=1pt,node distance=2cm,on grid,auto, initial text=]
      \node[state, initial, accepting] (s) {$s$};
    \end{tikzpicture}\hspace{17mm}
    \begin{tikzpicture}[shorten >=1pt,node distance=2cm,on grid,auto, initial text=]
      \node[state, initial] (s) {$s$};
      \node[state, accepting] (f) [right=of s] {$f$};
      \path[->] (s) edge node {$a$} (f);
    \end{tikzpicture}
  \caption{The atomic automata for $\varnothing$, $\{\varepsilon\}$, and
  $\{a\}$, in this order.}
  \label{fig:atomic-automata}
\end{figure}
The first automaton may be assigned either parity. Color the only state of
each of the first two automata by~$0$, and color $s$ by $0$ and $f$ by
$1$ in the third. Their transition graphs have at most one edge, and deleting
the initial state leaves at most one component. All conditions in
\cref{def:compAuto} are immediate.
\end{proof}

\subsection{Closure under union}

Let $p\in\{0,1\}$, and let
$A=(Q_A,\Sigma,\Delta_A,s_A,F_A)$ and
$B=(Q_B,\Sigma,\Delta_B,s_B,F_B)$ be strongly cactus automata of
parity~$p$. Take disjoint copies of $A$ and $B$, identify their initial
states, and call the identified state $s_\cup$. Keep all transitions and
accepting states. Thus $s_\cup$ is accepting if and only if $s_A\in F_A$ or
$s_B\in F_B$. Denote the resulting automaton by~$U$.

\begin{lemma}\label{lem:union}
The automaton $U$ constructed above is strongly cactus of parity $p$ and
satisfies $L(U)=L(A)\cup L(B)$.
\end{lemma}

\begin{proof}
The empty word is accepted precisely when it is accepted by $A$ or by~$B$.
Every non-empty run of $U$ chooses, in its first step, a component from one of
the two automata. Since no transition enters $s_\cup$ and there are no
transitions between different components, the run stays on that side. Hence
$L(U)=L(A)\cup L(B)$.

The two colorings agree at their initial states, which both have color~$0$,
and therefore combine into a transition-compatible coloring of~$U$. All
accepting states have color~$p$. The components of $G(U)-s_\cup$ are exactly
the components of $G(A)-s_A$ and $G(B)-s_B$, copied disjointly. Their
entries, accepting states, and exposed vertices are unchanged. The two graphs
meet only in $s_\cup$, so their union is a cactus by
\cref{lem:cactus-operations}. No transition enters the identified initial
state. Therefore $U$ is strongly cactus of parity~$p$.
\end{proof}

\subsection{Closure under concatenation}

Let $A=(Q_A,\Sigma,\Delta_A,s_A,F_A)$ and
$B=(Q_B,\Sigma,\Delta_B,s_B,F_B)$ be strongly cactus of parities $p$
and $q$, respectively. Put $r=p+q$. If $F_B$ is empty, then $B$ has no
component by Condition~3 of \cref{def:compAuto}. Since its transition graph
is connected, $B$ consists only of its non-accepting initial state, and
$L(B)=\varnothing$. In this case the one-state automaton for the empty
language, assigned parity~$r$, recognizes $L(A)L(B)$. We therefore assume
that $F_B\neq\varnothing$.

\pagebreak
We construct an automaton $A\cdot B$ as follows.
\begin{enumerate}
    \item Introduce a fresh initial state $\dot s$ and declare it accepting
    if and only if $s_A\in F_A$ and $s_B\in F_B$.
    \item If $s_A\in F_A$, take copies of all components of $G(B)-s_B$,
    copy their internal transitions, and copy every transition from $s_B$ to
    an entry state as a transition from $\dot s$ to the corresponding copied
    entry state. Keep all copied accepting states. These components deal with
    an empty first factor.
    \item For every $g\in Q_A\setminus\{s_A\}$ with
    $\Gamma_A(g)\neq\varnothing$, let $C_g$ be the component of
    $G(A)-s_A$ containing~$g$. Take one copy $C_g'$ of this component and one
    copy $B^g$ of the whole automaton $B$, whose copied initial state is
    denoted by $s_B^g$. Copy all internal transitions of $C_g$ and all
    transitions of $B$, but erase the accepting marks in $C_g'$. Copy the
    ingoing transitions of $C_g$ as transitions from $\dot s$, and add
    $(g',a,s_B^g)$ for every $a\in\Gamma_A(g)$. Keep all accepting states of
    $B^g$.
    \item If $\Gamma_A(s_A)\neq\varnothing$, take one copy $B^{s_A}$ of $B$,
    with copied initial state $s_B^{s_A}$. For every
    $a\in\Gamma_A(s_A)$, add $(\dot s,a,s_B^{s_A})$, and keep all accepting
    states of this copy. Since no transition enters $s_A$, an accepting run
    of $A$ whose last transition has source $s_A$ consists of that single
    transition. The new transition therefore reads the entire one-letter
    first factor and places the run at the initial state of the copied
    automaton~$B$. One copy of $B$ suffices for all such labels.
\end{enumerate}
The original transitions from $g$ into accepting states of $A$ are retained
in $C_g'$, but their targets are no longer accepting there. Hence a run may
pass an earlier visit to an old accepting state and choose a later occurrence
of a possible last transition as the boundary between the two factors.

\begin{lemma}\label{lem:concat}
The automaton $A\cdot B$ constructed in Steps~1--4 is strongly cactus
of parity $r$ and satisfies $L(A\cdot B)=L(A)L(B)$.
\end{lemma}

\begin{proof}
By Step~1, $\varepsilon$ is accepted exactly when
$\varepsilon\in L(A)\cap L(B)$. Let $u\in L(A)$ and $v\in L(B)$. If
$u=\varepsilon$ and~$v$ is non-empty, the components from Step~2 simulate an
accepting run of $B$ on~$v$. If also $v=\varepsilon$, Step~1 applies. Suppose
that $u$ is non-empty, choose an accepting run of $A$ on~$u$,
and let $(g,a,f)$ be its last transition. If $g\neq s_A$, use the component
from Step~3 indexed by~$g$: simulate the run up to $g'$ in $C_g'$, take the
bridge labeled $a$ to $s_B^g$, and then simulate an accepting run of $B$ on
$v$. If $g=s_A$, the accepting run of $A$ has length one because no
transition enters~$s_A$. Step~4 uses its label $a$ to enter $s_B^{s_A}$ and
then simulates the run of~$B$. This also covers $v=\varepsilon$ when the
copied initial state of $B$ is accepting. Thus
$L(A)L(B)\subseteq L(A\cdot B)$.

Conversely, the empty run is accepting only in the case already handled by
Step~1. Let a non-empty accepting run of $A\cdot B$ be given. If it enters a
component from Step~2, erasing the copy indices gives an accepting run of~$B$,
while $s_A\in F_A$ supplies the empty first factor. Otherwise the run enters
a component from Step~3 or Step~4. All accepting marks in the copied
$A$-component were erased, so in Step~3 the run must use one of the bridge
transitions from $g'$ to~$s_B^g$. If this transition has label~$a$, the definition of~$\Gamma_A(g)$ gives a transition $(g,a,f)$ into an accepting state of~$A$.
Replacing the bridge by this transition and erasing the copy indices gives an
accepting run of $A$ on the prefix. The part after the bridge gives an
accepting run of $B$ on the suffix. In Step~4, the first transition
corresponds directly to a transition $(s_A,a,f)$ into an accepting state of
$A$, and therefore supplies a one-letter first factor. Hence every accepted
word belongs to $L(A)L(B)$.

We define the coloring of the new automaton. Color $\dot s$ by~$0$. In every
copied $A$-component, retain the coloring $\chi_A$. In every Step~2 copy,
retain $\chi_B$. This step can occur only when $p=0$, because $s_A$ has
color~$0$ and is accepting. In every Step~3 copy $B^g$, use the shifted
coloring $p+\chi_B$. The source $g'$ of the bridge has color $1-p$, whereas
$s_B^g$ has color~$p$, so the bridge changes color. Finally, Step~4 can occur
only when $p=1$: a transition from $s_A$ into an accepting state changes the
color from $0$ to~$1$. We therefore use the shifted coloring $1+\chi_B$ in
this copy. Every transition changes color, and all accepting states have
color $p+q=r$.

The local graph of a component from Step~2 is a copied local graph of~$B$.
For Step~3, it consists of a copied local graph of $A$ and a copy of~$G(B)$
joined by the bridge $\{g',s_B^g\}$. For Step~4, it is a copy of~$G(B)$ joined
to $\dot s$ by a bridge. It follows from \cref{lem:cactus-operations} that
every local graph, and then their union at $\dot s$, is a cactus. There is one
entry in every component, every component contains an accepting state of the
copied automaton $B$, and no transition enters $\dot s$.

It remains to verify that every non-initial possible source of a last transition is
exposed. Components from Step~2 inherit this property from~$B$. Consider a
Step~3 component indexed by~$g$. If $s_B^g$ is accepting, then $g'$ is the
source of a transition into an accepting state. The copied local graph of $A$
has a unique $\dot s$-$g'$ path, because $g$ is exposed in~$A$.

Every other possible source of a last transition lies in the copy $B^g$.
There is a unique path from $\dot s$ to $s_B^g$: it is the unique
$\dot s$-$g'$ path followed by the bridge $\{g',s_B^g\}$. If the source is
$s_B^g$, this is the required path. If it is a state $h\neq s_B^g$, then $h$ is
exposed in~$B$, and the unique $s_B^g$-$h$ path in the corresponding copied
local graph extends the path from $\dot s$ to $s_B^g$. Since the bridge is the
only connection between the two copied graphs, the resulting
$\dot s$-$h$ path is unique.

Finally, a Step~4 component is a copy of $G(B)$ attached to $\dot s$ by the
bridge $\{\dot s,s_B^{s_A}\}$. The same argument, now with this bridge as the
initial segment, shows that all non-initial possible sources of last transitions
in this component are exposed. Therefore $A\cdot B$ is strongly cactus of
parity~$r$.
\end{proof}

\subsection{Closure under Kleene star for even languages}

Let $A=(Q,\Sigma,\Delta,s,F)$ be strongly cactus of parity~$0$, and let
$\chi_A$ be a coloring witnessing this. Recall that for $v\in Q$,
$\Gamma_A(v)=\{a\in\Sigma\mid (v,a,f)\in\Delta\text{ for some }f\in F\}$.
Put $L_+=L(A)\setminus\{\varepsilon\}$. Our aim is to construct a strongly
cactus automaton $A^\star$ for $L(A)^*=L_+^*$. The equality holds because
empty factors can be removed from every product without changing the
resulting word.

The naive suppression of $\epsilon$-transitions does not guarantee to yield an automaton with a cactus transition graph, as the example depicted in \cref{fig:non-cactus-star} illustrates.
\begin{figure}[H]
    \centering
    \begin{subfigure}{.49\textwidth}
    \centering
    \begin{tikzpicture}[shorten >=1pt, node distance=1cm, on grid, auto, initial text=]
        \node[state, initial, minimum size=.5cm]   (s)         {$s$};
        \node[state, minimum size=.5cm]            (a)    [right = of s]    {};
        \node[state, minimum size=.5cm]            (b)    [right = of a]  {};
        \node[state, minimum size=.5cm]            (c)    [above right = of b]  {};
        \node[state, accepting, minimum size=.5cm]            (d)    [right = of c]  {};
        \node[state, minimum size=.5cm]            (e)    [below right = of b]  {};
        \node[state, accepting, minimum size=.5cm]            (f)    [right = of e]  {};

        \path[->]
            (s)     edge    node    {} (a)
            (a)     edge    node    {} (b)
            (b)     edge    node    {} (c)
            (b)    edge    node    {} (e)
            (c)    edge    node    {} (d)
            (e)    edge    node    {} (f)
            ;
    \end{tikzpicture}
    \caption{A strongly cactus automaton of parity~$0$}
    \end{subfigure}
    \begin{subfigure}{.49\textwidth}
    \centering
        \begin{tikzpicture}[shorten >=1pt, node distance=1cm, on grid, auto, initial text=]
        \node[state, initial, accepting, minimum size=.5cm]   (s)         {$s$};
        \node[state, minimum size=.5cm]            (a)    [right = of s]    {};
        \node[state, minimum size=.5cm]            (b)    [right = of a]  {};
        \node[state, minimum size=.5cm]            (c)    [above right = of b]  {};
        \node[state, accepting, minimum size=.5cm]            (d)    [right = of c]  {};
        \node[state, minimum size=.5cm]            (e)    [below right = of b]  {};
        \node[state, accepting, minimum size=.5cm]            (f)    [right = of e]  {};

        \path[->]
            (s)     edge    node    {} (a)
            (a)     edge    node    {} (b)
            (b)     edge    node    {} (c)
            (b)    edge    node    {} (e)
            (c)    edge    node    {} (d)
            (d.north)    edge  [bend right]   node  [below]  {$\epsilon$} (s)
            (e)    edge    node    {} (f)
            (f.south)    edge  [bend left]  node    {$\epsilon$} (s)
            ;
    \end{tikzpicture}
    \caption{The same automaton with $\epsilon$-transitions}
    \end{subfigure}
    \begin{subfigure}{.49\textwidth}
    \centering
        \begin{tikzpicture}[shorten >=1pt, node distance=1cm, on grid, auto, initial text=]
        \node[state, initial, accepting, minimum size=.5cm]   (s)         {$s$};
        \node[state, minimum size=.5cm]            (a)    [right = of s]    {};
        \node[state, minimum size=.5cm]            (b)    [right = of a]  {};
        \node[state, minimum size=.5cm]            (c)    [above right = of b]  {};
        \node[state, accepting, minimum size=.5cm]            (d)    [right = of c]  {};
        \node[state, minimum size=.5cm]            (e)    [below right = of b]  {};
        \node[state, accepting, minimum size=.5cm]            (f)    [right = of e]  {};

        \path[->]
            (s)     edge    node    {} (a)
            (a)     edge    node    {} (b)
            (b)     edge    node    {} (c)
            (b)    edge    node    {} (e)
            (c)    edge    node    {} (d)
            (d.north)    edge  [bend right]   node  [above]  {} (a)
            (e)    edge    node    {} (f)
            (f.south)    edge  [bend left]  node    {} (a)
            ;
    \end{tikzpicture}
    \caption{Suppressing $\epsilon$-transitions}
    \end{subfigure}
    \begin{subfigure}{.49\textwidth}
    \centering
        \begin{tikzpicture}[shorten >=1pt, node distance=1cm, on grid, auto, initial text=]
        \node[state, initial, accepting, minimum size=.5cm]   (s)         {$s$};
        \node[state, minimum size=.5cm]            (a)    [right = of s]    {};
        \node[state, minimum size=.5cm]            (b)    [right = of a]  {};
        \node[state, minimum size=.5cm]            (c)    [above right = of b]  {};
        \node[state, accepting, minimum size=.5cm]            (d)    [right = of c]  {};
        \node[state, minimum size=.5cm]            (e)    [below right = of b]  {};
        \node[state, accepting, minimum size=.5cm]            (f)    [right = of e]  {};

        \path[->]
            (s)     edge    node    {} (a)
            (a)     edge    node    {} (b)
            (b)     edge    node    {} (c)
            (b)    edge    node    {} (e)
            (c)    edge    node    {} (d)
            (e)  [] edge   node    {} (f)
            (e.south west) [bend left]    edge    node    {} (s)
            (c.north west) [bend right]   edge    node    {} (s)
            ;
    \end{tikzpicture}
    \caption{Alternative for suppressing $\epsilon$-transitions}
    \end{subfigure}
    \caption{Example of naive $\epsilon$ suppression}
    \label{fig:non-cactus-star}
\end{figure}
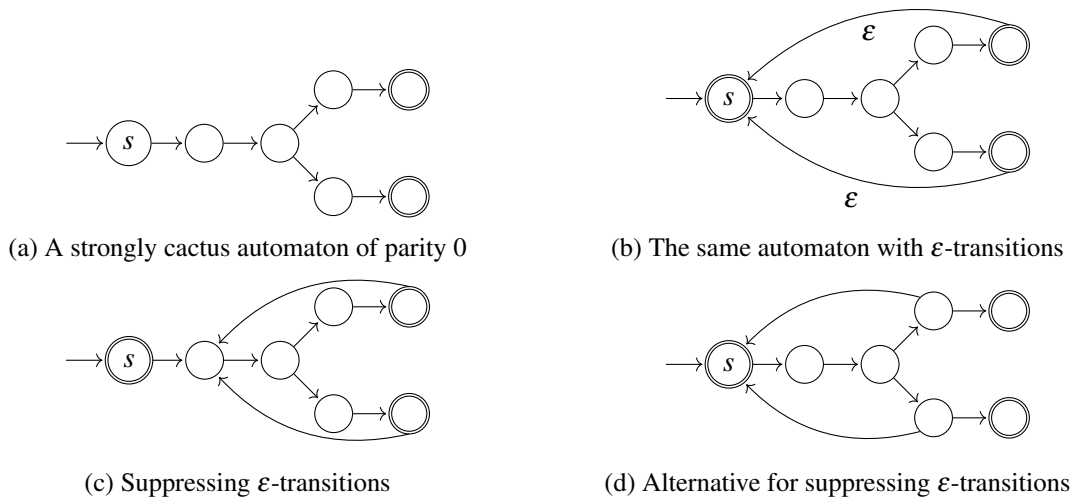

We make two observations. 
First, $\Gamma_A(s)=\varnothing$. Indeed,
both $s$ and every accepting state have color~$0$, whereas every transition
must change color. Thus the source of the last transition of a non-empty
accepting run is never~$s$.
Put $P_A=\{g\in Q\setminus\{s\}\mid\Gamma_A(g)\neq\varnothing\}$.
Hence $P_A$ is the set of non-initial states from which a transition into an accepting state
is available.
In particular, the source of the last transition of every non-empty accepting run belongs to~$P_A$.
Second, no transition enters $s$. Once a run leaves $s$, it therefore cannot
return to $s$ and remains in one component of $G(A)-s$. 

For $g\in P_A$, let
$C_g$ be the component of $G(A)-s$ containing $g$, and let $t_g$ be its
unique entry state. A run whose last-transition source is $g$ uses only $s$
and states of $C_g$. For $g\in P_A$, let $L_g$ consist of all words $ua$ such that $A$ has a run
from $s$ to $g$ on~$u$ and $a\in\Gamma_A(g)$. Then
$L_+=\bigcup_{g\in P_A}L_g$. Indeed, if $(g,a,f)$ is the last transition of a
non-empty accepting run, then $g\in P_A$ and the word read by the run belongs
to $L_g$. Conversely, the definition of $\Gamma_A(g)$ shows that every word
in $L_g$ is accepted by $A$.
Different states $g,h\in P_A$ may lie in the same component of $G(A)-s$. In
that case $C_g=C_h$.
The copies indexed by $g$ and $h$ below are nevertheless fresh and disjoint. 

We now explain the construction in detail.  
First, take a copy $\hat A$ of the entire automaton with states $\hat v \in \hat Q$ for every $v\in Q$.
The copied states retain their accepting marks and their coloring.
For every $g\in P_A$, we construct a new gadget $B_g$ as follows.

Take fresh copies $Q^g, \Delta^g$ of all states and transitions of $A$.
If $s\in F$, remove the accepting mark from the copied state $s^g$.
All other copied states retain their accepting mark.
We call the copied components of $G(A)-s$ \emph{terminal modules} of $B_g$.
For each $h\in P_A$ let $C_h^{R,g \to h}$ be fresh copies of the states and transitions of the component $C_h$.
Each copied state of $C_h^{R,g\to h}$ retains the color of the corresponding original state.
We call $C_h^{R,g \to h}$ \emph{return modules} of $B_g$.
All states in $C_h^{R, g\to h}$ have their accepting marks erased.
Denote $s^g$ as the \emph{restart state} of $B_g$ and attach each return module via the copied ingoing transitions $(s^g, a, t_h^{R,g \to h})$, for each $a\in \Lambda_A(s,t_h)$.

Now, for each $h\in P_A$ add return transitions $(h^{R,g \to h}, a, s^g)$ for $a\in \Gamma_A(h)$.
These transitions ensure that a run can continue after reading a word $w\in L_+$.
Attach each $B_g$ to $\hat A$ via bridge transitions $(\hat g, a,s^g)$ for $a\in \Gamma_A(g)$.
Finally, make $\hat s$ accepting, to account for the empty word.
This is permissible, since no transitions enters $\hat s$.
Let the resulting automaton be denoted by $A^\star = (Q^\star, \Sigma, \Delta^\star, \hat s, F^\star)$.
For all copied states $v'\in Q^\star$ with corresponding original state $v\in Q$ define $\chi_{A^\star}(v') := \chi_A(v)$.

The coloring defined in the construction is transition-compatible for $A^\star$.
This coloring is transition compatible on states $\hat v\in \hat Q$.
For $g\in P_A$ it holds that $\chi_A(g) = \chi_{A^\star}(\hat g) = 1$, since $A$ is of parity~0.
Thus $\chi_{A^\star}(s^g) = 0$ respects transition-compatibility. 
Similarly, for $h\in P_A$, $\chi_A(h) = \chi_{A^\star}(h^{R,g \to h}) = 1$, so \mbox{$\chi_{A^\star}(s^g)= 0$} remains transition-compatible.
The transition-compatibility of the coloring for all other states and transitions follows suit.

This replaces the usual $\varepsilon$-transitions in the standard star construction. 
We cannot return to the genuine initial state $\hat s$, because no transition may enter the initial state of a strongly cactus automaton. 
Note why we cannot make the restart state $s^g$ accepting.
Otherwise, the sources of the return transitions would become possible sources of last transitions, even though they may lie on the cycles closed by those return transitions.
This would violate the exposed-source condition.

\Cref{fig:star-two-sources} illustrates the construction when $G(A)-s$ has
only one component and two states $g,h\in P_A$, with transitions
$(g,a,f_g)$ and $(h,b,f_h)$ into accepting states. The figure shows the
gadget $B_g$. There is a separate gadget $B_h$ for a first factor whose
last-transition source is~$h$.

\begin{lemma}\label{lem:positive-gadget}
Let $L_+^+=\bigcup_{k\geq 1}L_+^k$. The words read by runs of $A^\star$ that start at
$s^g$ and end in an accepting state of $B_g$ are precisely
the words in~$L_+^+$. Moreover, the transition graph of
$B_g$ is a cactus and for every terminal module, every source of
a transition into a copied accepting state is joined to $s^g$ by a unique
path.
\end{lemma}
\begin{figure}[t]
\centering
\resizebox{\textwidth}{!}{
\begin{tikzpicture}[
    ->=Stealth,
    every node/.style={font=\small},
    vertex/.style={circle,draw,minimum size=6.5mm,inner sep=1pt},
    accepting/.style={vertex,double,double distance=1.2pt},
    module/.style={draw,rounded corners,minimum width=42mm,minimum height=9mm,
                   align=center,fill=gray!7}
]
    \node[font=\bfseries] at (1.8,5.2) {(a) The original automaton $A$ with component $C$};
    \node[vertex] (s) at (-0.2,3.2) {$s$};
    \node[vertex] (t) at (0.9,3.2) {$t$};
    \node[vertex] (og) at (2.1,3.8) {$g$};
    \node[vertex] (oh) at (2.1,2.6) {$h$};
    \node[accepting] (ofg) at (3.45,3.8) {$f_g$};
    \node[accepting] (ofh) at (3.45,2.6) {$f_h$};
    \draw[rounded corners] (0.41,2.05) rectangle (3.9,4.35);
    \node[anchor=north west] at (0.52,4.30) {$C$};
    \draw[->] ($(s.west) - (.4,0)$) to (s.west);
    \draw[->] (s) -- (t);
    \draw[->] (t) -- (og);
    \draw[->] (t) -- (oh);
    \draw[->] (og) -- node[above] {$a$} (ofg);
    \draw[->] (oh) -- node[below] {$b$} (ofh);

    \node[font=\bfseries] at (9.55,5.2) {(b) The gadget $B_g$ of $A^\star$};
    \node[accepting] (hs) at (5.05,3.2) {$\hat s$};
    \node[vertex] (ht) at (6.15,3.2) {$\hat t$};
    \node[vertex] (hg) at (7.1,3.8) {$\hat g$};
    \node[vertex] (hh) at (7.1,2.6) {$\hat h$};
    \node[accepting] (hfg) at (8.25,4.2) {$\hat f_g$};
    \node[accepting] (hfh) at (8.25,2.2) {$\hat f_h$};
    \draw[->] (hs) -- (ht);
    \draw[->] (ht) -- (hg);
    \draw[->] (ht) -- (hh);
    \draw[->] (hg) -- node[above,sloped] {$a$} (hfg);
    \draw[->] (hh) -- node[below,sloped] {$b$} (hfh);

    \node[vertex] (sg) at (9.85,3.2) {$s^g$};
    \draw[->,very thick] (hg) -- node[above,sloped] {$a$} (sg);
    
    \draw[rounded corners] (9.25,1.1) rectangle (14.95,4.65);
    \node[anchor = north west] at (9.3, 4.6) {$B_g$};
    \node[module] (rgg) at (12.75,4.075) {};
    \node[]  at  (12.75,4.075) {$C_{g}^{R,g \to g}:\ s^g\leadsto g^{R, g\to g}\xrightarrow{a}s^g$};
    \node[module] (rgh) at (12.75,2.875) {};
    \node[]  at  (12.75,2.875)    {$C_{h}^{R,g \to h}:\ s^g\leadsto h^{R, g\to h}\xrightarrow{b}s^g$};
    \node[module] (tgg) at (12.75,1.675) {};
    \node[]  at  (12.55,1.675)    {$C^{g}:\ s^g\leadsto f_h^{g} \text{ or } s^g\leadsto f_g^{g}$};

    \draw[->] ($(hs.west) - (.4,0)$) to (hs.west);
    \draw[->] (sg) to[bend left=17] (rgg.west);
    \draw[->,very thick] (sg) -- (rgh.west);
    \draw[->,very thick] (sg) to[bend right=9] (tgg.west);

    \node[module,dashed,minimum width=29mm] (bh) at (6.65,0.72)
        {symmetric gadget $B_h$};
    \draw[->,dashed, thick] (hh) to[bend right=24] node[above,sloped] {$b$} (bh.north west);

    \node[align=center,anchor=north] at (12.55,0.67)
        {all four choices are available whenever the run is at $s^g$};
\end{tikzpicture}
}
\caption{A schematic example with one original component and two possible
last-transition sources $g$ and $h$. The thick transition completes the first
factor and reaches the private restart state $s^g$. The other thick arrows
illustrate a possible continuation whose next factor has last-transition
source $h$ and whose final factor has last-transition source $g$. After every
return, all four modules are available again. The component $B_h$, used when
the first factor ends through $h$, is shown only schematically. Only the
relevant vertices and transitions are shown.}
\label{fig:star-two-sources}
\end{figure}
\begin{proof}
Consider a traversal of a return module $C^{R, g \to h}$. Replace its final
transition $(h^{R, g\to h},a,s^g)$ by any transition $(h,a,f)$ witnessing
$a\in\Gamma_A(h)$, replace $s^g$ by $s$, and erase the copy indices. This
gives an accepting run of~$A$ on the same word. Conversely, every run
witnessing that a word belongs to $L_h$ can be copied into $C^{R, g\to h}$ and
completed by the corresponding transition back to~$s^g$. Thus the return
modules together read precisely the words in $L_+$ and return to~$s^g$.

The terminal modules together with $s^g$ are an unmodified copy of $A$.
As such, the terminal modules together read precisely the words in $L_+$ and end in an accepting state.
The modules meet only at $s^g$, and the only transitions entering $s^g$ are the final transitions of
return modules.
Since $s^g$ is not accepting, every accepting run consists of zero or more traversals of return modules followed by one traversal of a terminal module.
It therefore reads a word in~$L_+^+$.

Conversely, let $w=w_1\cdots w_k$, where $k\geq 1$ and $w_i\in L_+$ for every
$i$. For every $i$, choose $h_i\in P_A$ with $w_i\in L_{h_i}$. For $i<k$,
read $w_i$ in the return module $C^{R, g \to h_i}$.
After each of these factors the run is again at $s^g$, so the choice of $h_{i+1}$ is independent of~$h_i$.
Finally, read $w_k$ in the terminal module $C_{h_k}^{g}$, which contains a copy of the chosen accepting run of $A$ on~$w_k$.
This gives an accepting run of the gadget on~$w$.

It remains to verify the graph properties.
For $h\in P_A$, put $X_h=G(A)[V(C_h)\cup\{s\}]$.
The local graph of a terminal module is a copy of $X_h$, with $s$ replaced by~$s^g$, and is therefore a cactus.
Let $v^g$ be the source of a transition into a copied accepting state in this module.
Then $v$ is exposed in~$A$, so $X_h$ has a unique $s$-$v$ path.
Its copy is the unique $s^g$-$v^{g}$ path in the terminal module.

The local graph of the return module $C_v^{R,g\to h}$ is obtained from a copy of
$X_h$ by adding the edge $\{s^g,h^{R, g\to h}\}$.
Since $h$ is exposed in $A$, the graph $X_h$ has a unique $s$-$h$ path.
Hence adding this edge preserves the cactus property by \cref{lem:cactus-operations}: it creates at
most one new cycle.

All modules meet only at~$s^g$, so their union is again a cactus.
Attaching the other modules at $s^g$ cannot create another simple path from $s^g$ to a vertex of a terminal module, because such a path would have to leave $s^g$ through another module and later return to~$s^g$.
Thus the path to every possible last-transition source in a terminal module remains unique.
\end{proof}

In the first copy $\hat A$ two things can happen. A run may simply end in one of the
accepting states retained in this copy. In that case it reads one non-empty
factor from~$L(A)$. Alternatively, it may use one of the new transitions from
$\hat g$ to~$s^g$. The run then starts with a copied transition
$(\hat s,b,\hat t_g)$, simulates a run of $A$ up to~$\hat g$, and takes
$(\hat g,a,s^g)$. This last transition reads the same letter $a$ as a possible
final transition $(g,a,f)$ of~$A$. Hence the whole passage from $\hat s$ to
$s^g$ reads exactly one word from~$L_g$. 

Once the run is at $s^g$, it starts the next factor in exactly the same way as $A$ starts at~$s$: a copied transition enters either a return module $(s^g, a, t_h^{R, g\to h})$ or a terminal module $(s^g, a, t_h^g)$, for $a \in \Lambda_A(s,t_h)$.
At the end of a middle factor, a transition $(h^{R,g \to h}, b,s^g)$ reads its final letter and brings the run back to the same state~$s^g$.
The run can repeat this step as often as desired.
For the last factor it chooses a terminal module and ends in one of the accepting states retained there.

The states $s^g$ are private. If $g\neq h$, then $s^g\neq s^h$, even if
$C_g=C_h$. If two components shared the same restart state, then deleting
$\hat s$ would merge them and the resulting component would have more than
one entry from~$\hat s$. There would also be different paths from $\hat s$ to
the common restart state, which would destroy the unique-path property needed
for the terminal modules.

The automaton $A^\star$ consists exactly of a copy $\hat A$ of $A$ and the components $B_g$
constructed above.

\begin{lemma}\label{lem:kleene}
The automaton $A^\star$ is strongly cactus of parity~$0$ and satisfies
$L(A^\star)=L(A)^*$.
\end{lemma}

\begin{proof}
We first verify the accepted language. The empty run at $\hat s$ accepts
$\varepsilon$. The runs that remain in the first copy $\hat A$ of $A$
together accept precisely the words in~$L_+$. Indeed, each such run is a
copied non-empty accepting run of~$A$. Conversely, if $A$ has a non-empty
accepting run and $g\in P_A$ is the source of its last transition, then the run can
be copied into the $\hat A$.

If a run uses the transition from $\hat g$ to $s^g$, its prefix up to $s^g$
reads a word from~$L_g$. By \cref{lem:positive-gadget}, the remaining part
reads a word from~$L_+^+$. Hence the runs that use a restart state together
accept precisely $(\bigcup_{g\in P_A}L_g)L_+^+=L_+L_+^+$, the products of at
least two words from~$L_+$. Consequently
$L(A^\star)=\{\varepsilon\}\cup L_+\cup L_+L_+^+=L_+^*=L(A)^*$.

This also covers $P_A=\varnothing$. Every non-empty accepting run of $A$ has
a last transition whose source belongs to $P_A$. Thus $P_A=\varnothing$
implies $L_+=\varnothing$, and our construction leaves only the accepting
initial state~$\hat s$.

We next verify the graph conditions. For $g\in P_A$, recall that
$X_g=G(A)[V(C_g)\cup\{s\}]$. The graph formed by $\hat s$ and $\hat C_g$ is a
copy of~$X_g$, where $\hat C_g \subset \hat Q$ is the copy of $C_g$.
The edge $\{\hat g,s^g\}$ joins this graph to the cactus $G(B_g)$.
It is a bridge because the two graphs are otherwise disjoint.
Hence the local graph of $B_g$ is a cactus.
All local graphs meet only at~$\hat s$, so $G(A^\star)$ is a connected cactus by \cref{lem:cactus-operations}.

Deleting $\hat s$ separates all components of $A^\star$ from one another.
A component $C^\star_g$ of $G(A^\star)-\hat s$ is of the form $\hat C_g\cup \bigcup_{h\in P_A\cap V(C_g)} G(B_h)$.
The unique entry of $C^\star_g$ is~$\hat t_g$.
It contains the accepting states retained in $\hat C_g$ as well as the accepting states of its terminal modules.
No transition enters~$\hat s$; return transitions enter only the private states~$s^g$.

It remains to verify the exposed-source condition. First let $\hat h$ be the
source of a transition into an accepting state retained in~$\hat C_g$. Then $h$
is exposed in~$A$, so $X_g$ has a unique $s$-$h$ path. Its copy is the unique
$\hat s$-$\hat h$ path in the first part of~$C^\star_g$. The bridge
$\{\hat g,s^g\}$ and the graph attached beyond it cannot create a second simple
path to~$\hat h$.

Now let $v^{g}$ be the source of a transition into an accepting state
of a terminal module. Since $g$ is exposed in~$A$, the first copy has a unique
$\hat s$-$\hat g$ path. Together with the bridge $\{\hat g,s^g\}$, this is the
unique $\hat s$-$s^g$ path. By \cref{lem:positive-gadget}, the vertex
$v^{g}$ is joined to $s^g$ by a unique path. Concatenating the two
paths gives the unique path from $\hat s$ to~$v^{g}$.

Finally, retain the coloring $\chi_A$ in every copied component and color
$\hat s$ and all private states $s^g$ by~$0$. Every copied accepting state has
color~$0$. Every entry state has color~$1$, because it is the endpoint of a
transition from $s$, which has color~$0$. Moreover, every $g\in P_{A}$ has
color~$1$, because it is the source of a transition into an accepting state
of color~$0$. Consequently every new transition changes color, and all copied
internal transitions change color because $\chi_A$ is
transition-compatible. We have verified all conditions of
\cref{def:compAuto}. Therefore $A^\star$ is strongly cactus of parity~$0$.
\end{proof}

\subsection{Proof of the main theorem}

We first do the typed induction.

\begin{theorem}\label{thm:typed-automata}
Let $p\in\{0,1\}$ and let $r$ be a regular expression typed with parity~$p$.
Then there exists an NFA accepting~$L(r)$ that is strongly
cactus of parity~$p$. In particular, its transition graph is a bipartite
cactus.
\end{theorem}

\begin{proof}
Proceed by structural induction on the typed expression~$r$. The atomic cases
are supplied by \cref{lem:atomic}. If $r=u\cup v$, the two subexpressions have
the same parity, and \cref{lem:union} applies. If $r=uv$, apply
\cref{lem:concat}. The parity of the result is the sum of the two input
parities. Finally, if $r=u^*$, the typing rules guarantee that $u$ has
parity~$0$, and \cref{lem:kleene} applies. 
\end{proof}

Then we put together the typed automata.

\begin{proof}[Proof of \cref{the:main}]
Let $r$ be a regular expression for the given language~$L$. By
\cref{lem:parity-normal-form}, there are expressions $r_0$ and $r_1$ typed
with parity $0$ and $1$ whose languages are the even and odd parts of~$L$.
By \cref{thm:typed-automata}, choose strongly cactus automata $A_0$ and
$A_1$ for these expressions. Apply the union construction of
\cref{lem:union} to $A_0$ and~$A_1$. The resulting automaton recognizes
$L(r_0)\cup L(r_1)=L$ and has a bipartite cactus transition graph.
\end{proof}

\section{Conclusion}\label{sec:conclusion}

We proved that every regular language is recognized by an $\varepsilon$-free
NFA whose transition graph is a bipartite cactus. This is substantially more
restrictive than planarity, and also more restrictive than general
series-parallel structure: every cycle is even, and two cycles can meet only
at a cut vertex. Cactus graphs have nice structural properties, that make some generally intractable problems solvable in polynomial or even linear time.

The proof is constructive, but its normalization and closure steps duplicate
components and may introduce many auxiliary states. We have not optimized the
construction or established tight state-complexity bounds. Determining the
smallest possible blow-up, either as a function of the input automata or
relative to a minimal DFA, is therefore a natural remaining problem.

For deterministic automata, genus and crossing-number parameters offer
complementary ways to measure how far a regular language is from admitting a
planar deterministic representation. Clarifying the hierarchies induced by
these parameters, and comparing their behavior with the nondeterministic
bipartite-cactus representations proved here, would sharpen the topological
theory of regular languages.

Another direction is to identify which language-theoretic restrictions are
forced by planarity, outerplanarity, or the cactus property in deterministic,
unambiguous, and finitely ambiguous automata. It would also be useful to
determine which additional automata-theoretic properties can be imposed
simultaneously with a bipartite cactus transition graph while retaining a
representation theorem for all regular languages.

\bibliographystyle{eptcs}
\bibliography{resources.bib}

\begin{thebibliography}{10}
\providecommand{\bibitemdeclare}[2]{}
\providecommand{\surnamestart}{}
\providecommand{\surnameend}{}
\providecommand{\urlprefix}{Available at }
\providecommand{\url}[1]{\texttt{#1}}
\providecommand{\href}[2]{\texttt{#2}}
\providecommand{\urlalt}[2]{\href{#1}{#2}}
\providecommand{\doi}[1]{doi:\urlalt{https://doi.org/#1}{#1}}
\providecommand{\eprint}[1]{arXiv:\urlalt{https://arxiv.org/abs/#1}{#1}}
\providecommand{\bibinfo}[2]{#2}

\bibitemdeclare{article}{ArnborgLagergrenSeese1991}
\bibitem{ArnborgLagergrenSeese1991}
\bibinfo{author}{Stefan \surnamestart Arnborg\surnameend},
  \bibinfo{author}{Jens \surnamestart Lagergren\surnameend} \&
  \bibinfo{author}{Detlef \surnamestart Seese\surnameend}
  (\bibinfo{year}{1991}): \emph{\bibinfo{title}{Easy Problems for
  Tree-Decomposable Graphs}}.
\newblock {\slshape \bibinfo{journal}{Journal of Algorithms}}
  \bibinfo{volume}{12}(\bibinfo{number}{2}), pp. \bibinfo{pages}{308--340},
  \doi{10.1016/0196-6774(91)90006-K}.

\bibitemdeclare{misc}{BezakovaPal1999}
\bibitem{BezakovaPal1999}
\bibinfo{author}{Ivona \surnamestart Bez{\'a}kov{\'a}\surnameend} \&
  \bibinfo{author}{Martin \surnamestart P{\'a}l\surnameend}
  (\bibinfo{year}{1999}): \emph{\bibinfo{title}{Planar Finite Automata}}.
\newblock \bibinfo{howpublished}{Student Science Conference, Comenius
  University, Bratislava}.

\bibitemdeclare{inproceedings}{Bojanczyk2008Common}
\bibitem{Bojanczyk2008Common}
\bibinfo{author}{Miko{\l}aj \surnamestart Boja{\'n}czyk\surnameend}
  (\bibinfo{year}{2008}): \emph{\bibinfo{title}{The Common Fragment of {ACTL}
  and {LTL}}}.
\newblock In \bibinfo{editor}{Roberto~M. \surnamestart Amadio\surnameend},
  editor: {\slshape \bibinfo{booktitle}{Foundations of Software Science and
  Computational Structures}}, {\slshape \bibinfo{series}{Lecture Notes in
  Computer Science}} \bibinfo{volume}{4962}, \bibinfo{publisher}{Springer}, pp.
  \bibinfo{pages}{172--185}, \doi{10.1007/978-3-540-78499-9_13}.

\bibitemdeclare{article}{BonfanteDeloup2018}
\bibitem{BonfanteDeloup2018}
\bibinfo{author}{Guillaume \surnamestart Bonfante\surnameend} \&
  \bibinfo{author}{Florian \surnamestart Deloup\surnameend}
  (\bibinfo{year}{2018}): \emph{\bibinfo{title}{The Genus of Regular
  Languages}}.
\newblock {\slshape \bibinfo{journal}{Mathematical Structures in Computer
  Science}} \bibinfo{volume}{28}(\bibinfo{number}{1}), pp.
  \bibinfo{pages}{14--44}, \doi{10.1017/S0960129516000037}.

\bibitemdeclare{article}{BonfanteDeloup2019}
\bibitem{BonfanteDeloup2019}
\bibinfo{author}{Guillaume \surnamestart Bonfante\surnameend} \&
  \bibinfo{author}{Florian~L. \surnamestart Deloup\surnameend}
  (\bibinfo{year}{2019}): \emph{\bibinfo{title}{Decidability of Regular
  Language Genus Computation}}.
\newblock {\slshape \bibinfo{journal}{Mathematical Structures in Computer
  Science}} \bibinfo{volume}{29}(\bibinfo{number}{9}), pp.
  \bibinfo{pages}{1428--1443}, \doi{10.1017/S0960129519000057}.

\bibitemdeclare{article}{BonfanteDeloup2024}
\bibitem{BonfanteDeloup2024}
\bibinfo{author}{Guillaume \surnamestart Bonfante\surnameend} \&
  \bibinfo{author}{Florian~L. \surnamestart Deloup\surnameend}
  (\bibinfo{year}{2024}): \emph{\bibinfo{title}{The Genus of Regular Languages
  and Directed Graph Emulators}}.
\newblock {\slshape \bibinfo{journal}{Theoretical Computer Science}}
  \bibinfo{volume}{1000}, p. \bibinfo{pages}{114579},
  \doi{10.1016/j.tcs.2024.114579}.

\bibitemdeclare{article}{BookChandra1976}
\bibitem{BookChandra1976}
\bibinfo{author}{Ronald~V. \surnamestart Book\surnameend} \&
  \bibinfo{author}{Ashok~K. \surnamestart Chandra\surnameend}
  (\bibinfo{year}{1976}): \emph{\bibinfo{title}{Inherently Nonplanar
  Automata}}.
\newblock {\slshape \bibinfo{journal}{Acta Informatica}}
  \bibinfo{volume}{6}(\bibinfo{number}{1}), pp. \bibinfo{pages}{89--94},
  \doi{10.1007/BF00263745}.

\bibitemdeclare{book}{brandstadt1999graph}
\bibitem{brandstadt1999graph}
\bibinfo{author}{Andreas \surnamestart Brandst{\"a}dt\surnameend},
  \bibinfo{author}{Van~Bang \surnamestart Le\surnameend} \&
  \bibinfo{author}{Jeremy~P \surnamestart Spinrad\surnameend}
  (\bibinfo{year}{1999}): \emph{\bibinfo{title}{Graph classes: a survey}}.
\newblock \bibinfo{publisher}{SIAM}, \doi{10.1137/1.9780898719796}.

\bibitemdeclare{article}{BRIMKOV2017393}
\bibitem{BRIMKOV2017393}
\bibinfo{author}{Boris \surnamestart Brimkov\surnameend} \&
  \bibinfo{author}{Illya~V. \surnamestart Hicks\surnameend}
  (\bibinfo{year}{2017}): \emph{\bibinfo{title}{Memory efficient algorithms for
  cactus graphs and block graphs}}.
\newblock {\slshape \bibinfo{journal}{Discrete Applied Mathematics}}
  \bibinfo{volume}{216}, pp. \bibinfo{pages}{393--407},
  \doi{10.1016/j.dam.2015.10.032}.
\newblock \bibinfo{note}{Graph-theoretic and Polyhedral Combinatorics Issues
  and Approaches in Imaging Sciences}.

\bibitemdeclare{article}{Courcelle1990MSO}
\bibitem{Courcelle1990MSO}
\bibinfo{author}{Bruno \surnamestart Courcelle\surnameend}
  (\bibinfo{year}{1990}): \emph{\bibinfo{title}{The Monadic Second-Order Logic
  of Graphs. {I}. Recognizable Sets of Finite Graphs}}.
\newblock {\slshape \bibinfo{journal}{Information and Computation}}
  \bibinfo{volume}{85}(\bibinfo{number}{1}), pp. \bibinfo{pages}{12--75},
  \doi{10.1016/0890-5401(90)90043-H}.

\bibitemdeclare{article}{Eggan1963}
\bibitem{Eggan1963}
\bibinfo{author}{L.~C. \surnamestart Eggan\surnameend} (\bibinfo{year}{1963}):
  \emph{\bibinfo{title}{Transition Graphs and the Star-Height of Regular
  Events}}.
\newblock {\slshape \bibinfo{journal}{Michigan Mathematical Journal}}
  \bibinfo{volume}{10}(\bibinfo{number}{4}), pp. \bibinfo{pages}{385--397},
  \doi{10.1307/mmj/1028998975}.

\bibitemdeclare{book}{Eilenberg1974}
\bibitem{Eilenberg1974}
\bibinfo{author}{Samuel \surnamestart Eilenberg\surnameend}
  (\bibinfo{year}{1974}): \emph{\bibinfo{title}{Automata, Languages, and
  Machines, Volume A}}.
\newblock {\slshape \bibinfo{series}{Pure and Applied
  Mathematics}}~\bibinfo{volume}{59}, \bibinfo{publisher}{Academic Press},
  \bibinfo{address}{New York}.

\bibitemdeclare{book}{HopcroftMotwaniUllman2006}
\bibitem{HopcroftMotwaniUllman2006}
\bibinfo{author}{John~E. \surnamestart Hopcroft\surnameend},
  \bibinfo{author}{Rajeev \surnamestart Motwani\surnameend} \&
  \bibinfo{author}{Jeffrey~D. \surnamestart Ullman\surnameend}
  (\bibinfo{year}{2006}): \emph{\bibinfo{title}{Introduction to Automata
  Theory, Languages, and Computation}}, \bibinfo{edition}{3} edition.
\newblock \bibinfo{publisher}{Addison-Wesley}.

\bibitemdeclare{incollection}{Kleene1956}
\bibitem{Kleene1956}
\bibinfo{author}{Stephen~C. \surnamestart Kleene\surnameend}
  (\bibinfo{year}{1956}): \emph{\bibinfo{title}{Representation of Events in
  Nerve Nets and Finite Automata}}.
\newblock In \bibinfo{editor}{Claude~E. \surnamestart Shannon\surnameend} \&
  \bibinfo{editor}{John \surnamestart McCarthy\surnameend}, editors: {\slshape
  \bibinfo{booktitle}{Automata Studies}}, {\slshape \bibinfo{series}{Annals of
  Mathematics Studies}}~\bibinfo{volume}{34}, \bibinfo{publisher}{Princeton
  University Press}, \bibinfo{address}{Princeton, NJ}, pp.
  \bibinfo{pages}{3--42}, \doi{10.1515/9781400882618-002}.

\bibitemdeclare{book}{Kozen1997}
\bibitem{Kozen1997}
\bibinfo{author}{Dexter~C. \surnamestart Kozen\surnameend}
  (\bibinfo{year}{1997}): \emph{\bibinfo{title}{Automata and Computability}}.
\newblock \bibinfo{series}{Undergraduate Texts in Computer Science},
  \bibinfo{publisher}{Springer}, \bibinfo{address}{New York},
  \doi{10.1007/978-1-4612-1844-9}.

\bibitemdeclare{article}{KutribMalcherSchneider2022}
\bibitem{KutribMalcherSchneider2022}
\bibinfo{author}{Martin \surnamestart Kutrib\surnameend},
  \bibinfo{author}{Andreas \surnamestart Malcher\surnameend} \&
  \bibinfo{author}{Christian \surnamestart Schneider\surnameend}
  (\bibinfo{year}{2022}): \emph{\bibinfo{title}{Finite Automata with Undirected
  State Graphs}}.
\newblock {\slshape \bibinfo{journal}{Acta Informatica}}
  \bibinfo{volume}{59}(\bibinfo{number}{1}), pp. \bibinfo{pages}{163--181},
  \doi{10.1007/s00236-021-00402-0}.

\bibitemdeclare{article}{RabinScott1959}
\bibitem{RabinScott1959}
\bibinfo{author}{Michael~O. \surnamestart Rabin\surnameend} \&
  \bibinfo{author}{Dana \surnamestart Scott\surnameend} (\bibinfo{year}{1959}):
  \emph{\bibinfo{title}{Finite Automata and Their Decision Problems}}.
\newblock {\slshape \bibinfo{journal}{IBM Journal of Research and Development}}
  \bibinfo{volume}{3}(\bibinfo{number}{2}), pp. \bibinfo{pages}{114--125},
  \doi{10.1147/rd.32.0114}.

\bibitemdeclare{article}{Thompson1968RegularExpressionSearch}
\bibitem{Thompson1968RegularExpressionSearch}
\bibinfo{author}{Ken \surnamestart Thompson\surnameend} (\bibinfo{year}{1968}):
  \emph{\bibinfo{title}{Programming Techniques: Regular Expression Search
  Algorithm}}.
\newblock {\slshape \bibinfo{journal}{Communications of the ACM}}
  \bibinfo{volume}{11}(\bibinfo{number}{6}), pp. \bibinfo{pages}{419--422},
  \doi{10.1145/363347.363387}.

\end{thebibliography}
\end{document}